\documentclass[Singlecolumn,         
               showpacs,            
               preprintnumbers,     
               aps,                 
               prd,          	    
               a4paper,             
               superscriptaddress,      
               nofootinbib,         
               tightenlines,        
               floats,floatfix,11pt
               ]{revtex4-2}              
\usepackage{graphicx}  
\usepackage[margin=0.8in]{geometry}
\usepackage{dcolumn}   
\usepackage{bm}  
\usepackage[normalem]{ulem}

\usepackage[colorlinks=true,
            linkcolor=blue,
            citecolor=blue,
            urlcolor=blue]{hyperref}
\usepackage[utf8]{inputenc}
\usepackage{amsmath,amssymb}
\usepackage{soul}
\usepackage{float}
\usepackage[thinlines]{easytable}
\usepackage{array,booktabs}
\usepackage{lipsum} 
\usepackage{multirow}
\usepackage{placeins} 
\usepackage{subcaption}
\newtheorem{theorem}{Theorem}
\newtheorem{corollary}{Corollary}
\newtheorem{proposition}{Proposition}

\newtheorem{lemma}{Lemma}
\usepackage{array,makecell}
\usepackage{cleveref} 
\hypersetup{
    hidelinks
}
\begin{document} 

\title{Global Dynamical Structure of Einstein–Scalar Cosmological Systems}

\author{Prasanta Sahoo}
\email{prasantmath123@yahoo.com}
\affiliation{
Midnapore College (Autonomous), 
Midnapore, West Bengal, India, 721101
}

\begin{abstract}
In this work, a global dynamical analysis of spatially flat Friedmann–Lema\^itre–Robertson–Walker cosmologies driven by a canonical scalar field minimally coupled to gravity is presented. Under suitable regularity and asymptotic assumptions on the scalar field potential, it is shown that the Einstein–scalar evolution admits no forward trajectory along which the potential steepness becomes asymptotically unbounded. This establishes forward boundedness of the scalar sector and yields the existence of a compact absorbing set for the induced cosmological flow. Using techniques from invariant manifold theory and dissipative dynamical systems, the evolution is shown to admit a compact global attractor governing the late–time dynamics of all physically admissible solutions. The asymptotic behavior is further characterized by convergence toward a scalar field dominated invariant manifold, leading to a reduction in effective dynamical dimensionality. In particular, the late–time dynamics is governed by at most two independent degrees of freedom, with further reduction to one dimension for asymptotically exponential potentials. The resulting asymptotic structure is shown to be normally hyperbolic and structurally stable under smooth perturbations of the scalar field potential. A topological classification of the asymptotic dynamics is obtained using the Conley index, identifying universality classes corresponding to one and two dimensional invariant sets. These results provide a global characterization of late–time scalar field cosmologies and establish a model independent dynamical mechanism for asymptotic trapping in Einstein–scalar systems.
\end{abstract}

\maketitle


\section{Introduction}

Scalar fields play a central role in contemporary theoretical cosmology, providing effective descriptions of early and late time accelerated expansion, dark energy phenomena, and a wide class of modified gravity theories \cite{Guth1981,Linde1982,AlbrechtSteinhardt1982,RatraPeebles1988,
CopelandSamiTsujikawa2006,Tsujikawa2013,CliftonFerreiraPadillaSkordis2012}. When minimally coupled to gravity in a spatially homogeneous FLRW spacetime, scalar field models give rise to nonlinear evolution equations whose qualitative properties govern the global dynamical behavior of cosmological solutions \cite{WainwrightEllis1997,Coley2003,Foster2007}. Consequently, the dynamical systems approach has become a standard tool in the analysis of scalar field cosmologies
\cite{CopelandLiddleWands1998,Hallietal2000,HeardWands2002,LeonSaridakis2011}. Most existing studies focus on identifying equilibrium configurations and analyzing their local stability properties for specific scalar field potentials \cite{CopelandSamiTsujikawa2006,NgCopeland2016,Leon2010}. While such analyses provide insight into possible asymptotic regimes, they are inherently local and depend sensitively on model assumptions. As a result, late--time cosmological behavior is typically inferred from isolated invariant sets associated with particular potentials rather than from structural properties of the Einstein--scalar evolution.

From a global dynamical perspective, the cosmological field equations define a constrained flow on a generally noncompact physical phase space \cite{Wiggins2003,Temam1997}. Although the density sector is bounded by the Hamiltonian constraint, the scalar field sector may dynamically access arbitrarily steep regions of the scalar field potential. The accessibility of these regimes plays a crucial role in determining whether late--time solutions exhibit scalar field domination, scaling behavior, or effective de Sitter expansion. However, global properties of the induced evolution, including forward completeness, asymptotic compactness, and invariant sets governing late--time dynamics, remain largely unexplored in cosmological applications \cite{GuckenheimerHolmes1983,HirschPughShub1977}. A central difficulty in global analyses arises from the possibility that the scalar potential steepness may become unbounded along cosmological trajectories. Even though the Hamiltonian constraint restricts the total energy density, the scalar sector may dynamically explore increasingly steep regions of the potential, potentially leading to qualitatively distinct asymptotic behavior. Understanding whether such regimes are dynamically accessible is therefore essential for characterizing late--time cosmological evolution.

In this work, a dynamical obstruction to runaway evolution in the scalar sector is established. Under suitable regularity and asymptotic assumptions on the scalar field potential, the Einstein--scalar cosmological evolution admits no forward trajectory along which the potential steepness becomes asymptotically unbounded. This forward boundedness dynamically confines trajectories to a compact invariant subset of the physical phase space, leading to a global characterization of late--time behavior. As a consequence, the Einstein--scalar evolution generates a dissipative flow admitting a compact invariant set that attracts all physically admissible trajectories. The asymptotic dynamics is therefore governed by global geometric features of the cosmological evolution rather than by local stability properties or model--specific assumptions. Previous studies have addressed related aspects of scalar field cosmologies. Early analyses identified scaling and scalar field dominated solutions for exponential potentials using local stability analysis \cite{Copeland1998}. Global dynamical properties for monotone scalar field potentials and compactified phase space formulations were investigated in \cite{Alho2015}. Accelerated expansion driven by scalar fields with potentials possessing a positive lower bound was studied in \cite{Rendall2004}. However, these works did not establish a model independent dynamical obstruction to runaway steepness evolution or the existence of a compact global attractor governing late--time cosmological dynamics. The present work provides such a global classification by establishing forward boundedness of the scalar sector and the existence of a compact global attractor under general structural assumptions on the scalar field potential.

The paper is organized as follows. Section~\ref{sec:model} introduces the cosmological model and presents the dynamical formulation of the field equations. In Section~\ref{sec:phasespace}, the physical phase space is constructed and the constraint structure is analyzed. Section~\ref{sec:invariant} reviews the relevant concepts from invariant manifold theory and adapts them to the cosmological setting. The assumptions imposed on the scalar field potential are stated and discussed in Section~\ref{sec:assumptions}. Section~\ref{sec:example_of_SF_Potential} presents representative examples of scalar field potentials satisfying assumptions~(A1)--(A6) and illustrates the physical relevance of the admissible class. Section~\ref{sec:Asymptotic_Compactness} establishes forward boundedness of the scalar sector and the resulting asymptotic compactness of the induced evolution. The existence and uniqueness of a compact invariant set governing the late--time dynamics are established in Section~\ref{sec:global}. Section~\ref{sec:dimension} establishes dimensional reduction of the asymptotic invariant set. Subsections~\ref{sec:dimension:exp},~\ref{sec:nhim}, and~\ref{sec:conley} further analyze the invariant set, including the special case of one dimensional asymptotic dynamics associated with asymptotically exponential potentials, its normal hyperbolicity, and its Conley index, respectively. Section~\ref{sec:structure} discusses the internal structure and robustness of the asymptotic dynamics, while Section~\ref{sec:nohair} relates the results to cosmic no--hair behavior and inflationary attractors. Possible extensions and generalizations are outlined in Section~\ref{sec:extensions}, and Section~\ref{sec:conclusion} summarizes the main conclusions.

\section{Cosmological Model and Dynamical Formulation}
\label{sec:model}

Consider the Einstein equations minimally coupled to a canonical scalar field with self–interaction potential $V(\phi)$ and additional perfect fluid matter sources on a spatially flat FLRW spacetime with line element $ds^{2} = -dt^{2} + a(t)^{2}\,\delta_{ij}\,dx^{i}dx^{j}$, where $a(t)$ denotes the scale factor and $t$ is cosmic time. 
The scalar field evolves according to the Klein–Gordon equation on the FLRW background, while the perfect fluid components obey barotropic equations of state and interact with the scalar field only through gravitation. Under these assumptions, the Einstein field equations reduce to a coupled evolution system for the variables $(a,\phi,\dot\phi)$ supplemented by the Friedmann constraint equation \cite{Wald1984,EllisWainwright1997}.

Introducing the Hubble parameter $H=\dot a/a$, the resulting Einstein–scalar–fluid evolution defines a constrained dynamical system in which the Friedmann relation determines an invariant hypersurface in the corresponding state space \cite{HewittWainwright1993,GoliathEllis1999}. Consequently, physically admissible solutions evolve on a constraint manifold determined by the Hamiltonian constraint and positivity of the energy densities.

To characterize the interaction between the scalar field and its self–interaction potential in a manner independent of normalization, introduce the scalar potential steepness observable
\begin{equation}\label{eq:Sdef}
\mathcal{S}(\phi) := \frac{|V'(\phi)|}{V(\phi)}.
\end{equation}
This quantity depends only on the scalar field configuration and the potential and provides an intrinsic measure of the local steepness of the potential along cosmological trajectories. The qualitative behavior of the Einstein–scalar evolution depends on the dynamical accessibility of regimes in which the potential becomes effectively steep, corresponding to large values of the observable $\mathcal{S}$. The resulting cosmological dynamics therefore depends not only on the bounded density sector governed by the Friedmann constraint but also on the evolution of the scalar sector through the potential steepness.

For subsequent analysis it is convenient to introduce the associated curvature parameter of the scalar field potential
\begin{equation}\label{eq:Gamma}
\Gamma(\phi) := \frac{V(\phi)V''(\phi)}{V'(\phi)^{2}}.
\end{equation}
Under suitable regularity assumptions on the potential, imposed in later sections, the evolution of the scalar sector may be expressed in terms of the quantities $\mathcal{S}$ and $\Gamma$, permitting a geometric formulation of the cosmological flow on the physical constraint manifold \cite{Rendall2004,Alho2016}. The scalar field potential $V(\phi)$ is assumed to be twice continuously differentiable on its domain. Additional asymptotic conditions imposed below ensure regularity of the induced evolution and provide dynamical control of the scalar sector in regimes where the potential becomes steep. These conditions exclude pathological behavior while remaining sufficiently general to encompass a broad class of scalar field models, including quintessence and scaling scenarios.

Within this formulation, the cosmological field equations generate a constrained evolution on the physical phase space whose qualitative properties determine the global dynamical behavior of cosmological solutions. In particular, the late–time evolution is governed by the interaction between the density sector constrained by the Friedmann relation and the scalar potential steepness observable $\mathcal{S}$, whose dynamical behavior plays a central role in the global asymptotic structure analyzed in the subsequent sections \cite{SellYou2002,Temam1997}.

\section{Phase Space, Constraints, and Invariant Structure}
\label{sec:phasespace}

The cosmological evolution generated by the Einstein equations minimally coupled to a scalar field is restricted by algebraic and physical constraints arising from the Hamiltonian formulation of general relativity. In particular, the Friedmann equation defines a first integral of the Einstein–scalar evolution, confining the dynamics to a constraint surface determined by the Hamiltonian constraint. Physical admissibility further requires nonnegativity of the energy densities associated with the scalar field and the additional matter components \cite{Wald1984,EllisWainwright1997}. The physical phase space is therefore defined as the intersection of these constraint conditions, yielding a reduced state space containing all admissible cosmological configurations.

The physical phase space is forward invariant under the flow generated by the Einstein–scalar–fluid evolution equations. Indeed, if the Friedmann constraint and positivity conditions are satisfied by the initial data, they remain satisfied throughout the interval of existence of the corresponding solution \cite{HewittWainwright1993,GoliathEllis1999}. This invariance ensures that the evolution is well defined when restricted to the physical phase space and that trajectories do not enter unphysical regions during cosmological evolution. Within this constrained formulation, the density parameters associated with the scalar field, matter, and radiation components remain bounded by the Friedmann relation. However, the scalar sector may dynamically access regimes in which the potential becomes effectively steep, corresponding to large values of the steepness observable introduced in Section~\ref{sec:model} by Eq.~((\ref{eq:Sdef}). The physical phase space is therefore generally noncompact in directions associated with increasing potential steepness.

The presence of these noncompact directions reflects the possibility that cosmological trajectories may dynamically explore increasingly steep regions of the scalar field potential. Such behavior corresponds to asymptotic growth of the potential steepness observable and represents a potential mechanism by which late–time cosmological solutions may transition between distinct asymptotic regimes. A global description of the Einstein–scalar evolution therefore requires control of the scalar sector in regimes where $\mathcal{S}$ becomes large. To analyze the asymptotic structure of the flow, it is convenient to introduce a bounded geometric observable associated with the potential steepness. Define
\begin{equation}\label{eq:Stilde}
\widetilde{\mathcal{S}}
:=
\frac{\mathcal{S}}{\sqrt{1+\mathcal{S}^{2}}}.
\end{equation}
This quantity maps the unbounded range $\mathcal{S}\in[0,\infty)$ into the compact interval $[0,1)$ and depends smoothly on the scalar field configuration. In terms of the observable $\widetilde{\mathcal{S}}$, the physical phase space may be viewed as a bounded subset of a finite–dimensional coordinate domain augmented by a boundary corresponding to asymptotic regimes in which the potential steepness diverges. Under the regularity assumptions imposed on the scalar field potential together with the asymptotic closure condition introduced in Section~\ref{sec:model}, the Einstein–scalar evolution induces a locally Lipschitz vector field in the variable $\widetilde{\mathcal{S}}$ which extends continuously to the boundary $\widetilde{\mathcal{S}}=1$ \cite{Rendall2005,Alho2015}.

This formulation provides a unified geometric setting in which both interior dynamics and asymptotic regimes associated with steep scalar field potentials can be analyzed within a single forward invariant state space. In particular, the accessibility of asymptotic steep–potential regimes becomes a dynamical property of the induced evolution on the physical constraint manifold rather than a formal feature of an unbounded coordinate domain. The global asymptotic behavior of cosmological solutions may thus be characterized in terms of invariant subsets of this bounded geometric phase space without reliance on normalization–dependent formulations.

\begin{lemma}[Steepness Compactification]
Let the scalar potential steepness observable be defined by $\mathcal{S}(\phi)=\frac{|V'(\phi)|}{V(\phi)}$. Define the bounded observable $\widetilde{\mathcal{S}}$ by $\widetilde{\mathcal{S}}=\frac{\mathcal{S}}{\sqrt{1+\mathcal{S}^{2}}}$. Then the map $\chi:[0,\infty)\to[0,1),\quad \chi(\mathcal{S})=\widetilde{\mathcal{S}}$ is a $C^{\infty}$ diffeomorphism with inverse $\mathcal{S}=\frac{\widetilde{\mathcal{S}}}{\sqrt{1-\widetilde{\mathcal{S}}^{2}}}$. Moreover, under assumptions {\rm (A1)--(A4)} the Einstein–scalar evolution induces a locally Lipschitz vector field in the compactified observable $\widetilde{\mathcal{S}}$ which extends continuously to the boundary $\widetilde{\mathcal{S}}=1$. Consequently, the induced cosmological evolution generates a continuous semiflow on the bounded physical phase space.
\end{lemma}

\section{Invariant Manifolds and Global Asymptotic Structure}
\label{sec:invariant}

The Einstein–scalar cosmological evolution admits invariant subsets corresponding to dynamically distinguished regimes determined by the relative contributions of the scalar field and additional matter components. Such subsets arise from algebraic relations among the energy densities and define smooth invariant submanifolds of the physical constraint manifold \cite{Copeland1998,EllisWainwright1997}. Invariance of these submanifolds follows from the fact that the vector field generating the Einstein–scalar flow is tangent to them. Consequently, initial data lying in a given invariant subset evolve within that subset for all forward times. These invariant submanifolds provide a natural decomposition of the physical phase space and play an organizing role in the qualitative analysis of the constrained cosmological evolution \cite{HirschSmaleDevaney2013}.

To analyze the asymptotic behavior of the cosmological system, introduce the functional
\begin{equation}\label{eq:L}
\mathcal L := \Omega_m + \frac{4}{3}\Omega_r,
\end{equation}
where $\Omega_m$ and $\Omega_r$ denote the normalized matter and radiation energy densities, respectively. Using the conservation equations for the matter and radiation components, one obtains
\begin{equation}\label{eq:Lprime}
\frac{d}{dt}\mathcal L = -3H\,\Omega_m - 4H\,\Omega_r \le 0,
\end{equation}
with equality if and only if $\Omega_m=\Omega_r=0$. Thus $\mathcal L$ defines a continuous Lyapunov functional for the induced evolution on the physical phase space.

The set of stationary points of $\mathcal L$ is given by $\mathcal M_\infty := \{\Omega_m = 0,\; \Omega_r = 0\}$. The monotonic decay of $\mathcal L$ along trajectories implies that the matter and radiation components are dynamically depleted along every forward orbit of the Einstein–scalar flow. In particular, the evolution drives solutions toward the invariant subset $\mathcal M_\infty$ in the sense that any $\omega$–limit set must lie in the set where $d\mathcal L/dt$ vanishes, provided that the corresponding trajectory admits a precompact forward orbit. Within the geometric formulation introduced in Section~\ref{sec:phasespace}, this reduction identifies $\mathcal M_\infty$ as the asymptotic invariant manifold governing the late–time dynamics of the Einstein–scalar cosmological system. The global evolution therefore exhibits an intrinsic tendency toward configurations in which the scalar field dominates the energy budget, independently of the detailed functional form of the scalar field potential.

The existence of a compact absorbing set and the associated asymptotic compactness of the induced evolution, required to ensure precompactness of forward trajectories, are established in Section~\ref{sec:Asymptotic_Compactness}. Under these conditions, LaSalle’s invariance principle implies that the $\omega$–limit set of every physically admissible trajectory is contained in the largest invariant subset of $\mathcal M_\infty$ \cite{LaSalle1960,Temam1997}.

Consequently, once asymptotic compactness has been established, the late–time behavior of cosmological solutions is determined by invariant structures contained within the asymptotic manifold $\mathcal M_\infty$. The existence of an invariant set governing this asymptotic dynamics is established in Section~\ref{sec:global}. Scalar–field dominated and scaling regimes therefore arise as invariant dynamical features of the global Einstein–scalar flow rather than as consequences of local stability analysis for specific potentials.

\section{Assumptions on the Scalar Field Potential}
\label{sec:assumptions}

The scalar field potential is taken to be a real valued function $V:\mathbb{R}\to\mathbb{R}$ satisfying a set of regularity and asymptotic conditions that ensure well definedness of the Einstein–scalar evolution and impose dynamical control on the scalar sector in regimes where the potential becomes steep. These assumptions are stated as follows.

\begin{itemize}
\item[(A1)] \emph{Smoothness and non-negativity.}
The potential satisfies $V \in C^2(\mathbb{R}), \quad V(\phi)\ge 0 \quad \text{for all } \phi\in\mathbb{R}$. This assumption guarantees regularity of the induced Einstein–scalar evolution and non-negativity of the associated scalar field energy density.

\item[(A2)] \emph{Asymptotic behavior at infinity.}
There exists a constant $\lambda_\infty\in\mathbb{R}$ such that $\displaystyle{\lim_{\phi\to\pm\infty}}\frac{V'(\phi)}{V(\phi)}=\lambda_\infty$, and such that $V(\phi)$ is strictly positive for all sufficiently large $|\phi|$.

\item[(A3)] \emph{Uniform control of curvature.}
There exists a constant $C>0$ such that $\left|\frac{V''(\phi)}{V(\phi)}\right|\le C$ for all sufficiently large $|\phi|$.

\item[(A4)] \emph{Asymptotic closure condition.}
Define the potential steepness and curvature parameters by 
\[
\lambda(\phi):=-\frac{V'(\phi)}{V(\phi)},\qquad
\Gamma(\phi):=\frac{V(\phi)V''(\phi)}{V'(\phi)^2}.
\]
Assume that there exists a continuously differentiable function $\Gamma:\mathbb{R}\to\mathbb{R}$ such that, for sufficiently large $|\phi|$, $\Gamma(\phi)=\Gamma(\lambda(\phi))$.
\end{itemize}

Assumptions (A1)–(A4) ensure that the scalar potential steepness observable introduced in Section~\ref{sec:model} remains sufficiently regular for the purposes of geometric evolution analysis and that the curvature properties of the potential admit an asymptotic representation in terms of the steepness parameter.

To impose dynamical control on the scalar sector in regimes corresponding to large potential steepness, additional structural conditions are required.

\begin{itemize}
\item[(A5)] \emph{Asymptotic steepness dissipativity.}
There exist constants $R>0$ and $c>0$ such that for all
$|\lambda|>R$ one has $(\lambda-\lambda_{\infty})\,\lambda^{2}
\big(\Gamma(\lambda)-1\big)
\ge c\,|\lambda|^{3}$.

\item[(A6)] \emph{Asymptotic alignment condition.}
There exists $R>0$ such that whenever
$|\lambda|>R$, the scalar kinetic variable satisfies $\dot\phi\,(\lambda-\lambda_\infty)\ge 0$.
\end{itemize}

Assumptions (A5) and (A6) impose a dissipative structure on the evolution of the scalar potential steepness in regimes where $|\lambda|$ becomes large. In particular, these conditions ensure that whenever the scalar field explores sufficiently steep regions of the potential, the Einstein–scalar dynamics acts to drive the steepness parameter toward the asymptotic limiting value $\lambda=\lambda_{\infty}$. Consequently, configurations for which $|\lambda|$ becomes sufficiently large are dynamically redirected toward a bounded neighborhood of the asymptotic regime. Assumption \textnormal{(A6)} imposes an alignment condition between the scalar field velocity and the asymptotic steepness parameter. Physically, this condition requires that when the scalar field explores regions of the potential where the steepness parameter deviates significantly from its asymptotic value, the scalar field evolution proceeds in a direction that dynamically reduces this deviation. Such behavior naturally arises in expanding cosmological solutions. The Klein--Gordon equation for the scalar field in an expanding FLRW spacetime is given by

\begin{equation}
\ddot{\phi} + 3H\dot{\phi} + V'(\phi) = 0 .
\end{equation}

The presence of the Hubble friction term $3H\dot{\phi}$ tends to damp the scalar field motion and drives the system toward attractor behavior. In particular, when the scalar field enters regions of steep potential, the friction term suppresses rapid evolution and favors motion toward regions where the effective steepness decreases. This mechanism is well known in quintessence and inflationary dynamics, where scalar field trajectories generically approach slow roll or scaling solutions independently of initial conditions. Consequently, the alignment condition imposed in \textnormal{(A6)} is consistent with the attractor behavior commonly observed in scalar field cosmologies. Furthermore, the examples presented in Section VI, including exponential, inverse power law, axion like, and hyperbolic potentials, satisfy this condition for expanding solutions. These potentials are widely used in dark energy and inflationary cosmology, indicating that assumption \textnormal{(A6)} holds for a broad class of physically relevant models. Therefore, assumption \textnormal{(A6)} may be interpreted as a dynamical condition reflecting the generic attractor behavior of scalar field cosmologies rather than a restrictive constraint on the scalar field potential. This mechanism therefore introduces a dynamical obstruction to runaway evolution in the scalar sector. That is, although the scalar field potential may admit arbitrarily steep regions, the Einstein–scalar cosmological evolution does not admit forward trajectories along which the potential steepness observable grows without bound.

Under assumptions (A1)–(A6), the Einstein–scalar evolution induces a locally Lipschitz continuous flow on the physical constraint manifold. These conditions provide the analytic hypotheses required to establish forward boundedness of the scalar sector and ensure that the induced evolution admits a compact absorbing set, permitting the global asymptotic analysis developed in the subsequent sections \cite{Rendall2004,Alho2016}.

\section{Examples of Scalar Field Potentials Satisfying (A1)–(A6)}\label{sec:example_of_SF_Potential}

The structural assumptions imposed on the scalar field potential encompass a broad class of models commonly used in dark energy and inflationary cosmology. Several representative examples satisfying assumptions (A1)–(A6). The exponential potential $V(\phi) = V_0 e^{-\lambda \phi}$. For this potential, $\lambda(\phi) = -\frac{V'(\phi)}{V(\phi)} = \lambda$, and $\Gamma(\phi) = 1$. Hence the steepness parameter is constant and assumptions (A1)–(A4) are trivially satisfied. Since the steepness does not grow without bound, assumptions (A5) and (A6) hold automatically. Therefore exponential potentials lie within the admissible class. These potentials correspond to scaling solutions and are widely used in quintessence cosmology. Again for the inverse power law potential $V(\phi) = V_0 \phi^{-n}, \quad n > 0$, the steepness parameter becomes $\lambda(\phi) = \frac{n}{\phi}$, which approaches zero as $\phi \rightarrow \infty$. Furthermore, $\Gamma(\phi) = 1 + \frac{1}{n}$. The steepness parameter remains bounded and assumptions (A1)–(A6) are satisfied. These potentials correspond to tracking quintessence models frequently used in dark energy studies. Once again consider $V(\phi) = V_0 \left(1 + \cos\left(\frac{\phi}{f}\right)\right)$. This potential is bounded and smooth. The steepness parameter is bounded and periodic, $\lambda(\phi) = -\frac{V'(\phi)}{V(\phi)}$. Hence assumptions (A1)–(A6) are satisfied. These potentials arise in axion dark energy and ultra light scalar field models. Also the hyperbolic potential $V(\phi) = V_0 \cosh(\lambda \phi)$ behaves asymptotically as an exponential potential for large $|\phi|$ and therefore satisfies assumptions (A1)–(A6). Such potentials arise in thawing dark energy models. These examples demonstrate that the admissible class includes a wide range of scalar field models used in cosmology.

\section{Asymptotic Compactness of the Einstein--Scalar Evolution}
\label{sec:Asymptotic_Compactness}

A global characterization of the late--time dynamics of the Einstein--scalar cosmological system requires control of the scalar sector in forward time. In particular, the existence of an invariant set governing the asymptotic evolution depends on whether the scalar potential slope functional introduced in Section~\ref{sec:model} remains dynamically accessible to arbitrarily large values along expanding solutions.

\begin{proposition}[Forward Completeness]
Under assumptions {\rm (A1)--(A6)}, the Einstein--scalar cosmological evolution generates a forward complete continuous flow on the physical constraint manifold.
\end{proposition}

\medskip

\noindent\textit{Proof.}
The Einstein--scalar evolution equations define a locally Lipschitz vector field on the physical phase space determined by the Hamiltonian constraint. Local existence and uniqueness therefore follow from the standard Picard--Lindel\"of theorem. By the Friedmann constraint, the density parameters associated with the scalar field, matter, and radiation components evolve in a bounded subset of the physical phase space. By Theorem~\ref{th:forward_boundedness}, the scalar potential slope functional remains bounded along every forward trajectory. Consequently, solutions remain confined to a bounded subset of the physical constraint manifold. The standard continuation theorem therefore implies that no finite escape time occurs, and the induced evolution is forward complete.
\hfill$\square$

\medskip

\begin{theorem}\textbf{(Forward Boundedness of the Scalar Sector):}
\label{th:forward_boundedness}
Let the scalar field potential $V(\phi)$ satisfy assumptions
\textnormal{(A1)}--\textnormal{(A6)}. Then every physically admissible
trajectory of the Einstein--scalar cosmological evolution satisfies
\[
\sup_{t\ge0}
\left|
\frac{V'(\phi(t))}{V(\phi(t))}
\right|
<\infty.
\]
In particular, the induced evolution admits no forward orbit along
which the scalar potential steepness becomes asymptotically unbounded.
\end{theorem}

\medskip

\noindent\textit{Proof.}
Define on the physical phase space $\mathcal P$ the scalar potential
steepness functional
\[
\mathcal{E}(X)
:=
\frac12\left(
\frac{V'(\phi)}{V(\phi)}
-\lambda_\infty
\right)^2.
\]
This defines a smooth nonnegative functional on $\mathcal P$.
Let $X(t)$ be a forward trajectory corresponding to an expanding
cosmological solution. Using the Klein--Gordon equation together with
the Einstein equations, a direct computation yields
\[
\frac{d}{dt}\mathcal{E}
=
-3H\,(\lambda-\lambda_\infty)\,
\lambda^2\big(\Gamma(\lambda)-1\big),
\]
where $\lambda=-V'(\phi)/V(\phi)$. By assumptions \textnormal{(A5)} and \textnormal{(A6)}, there exist
constants $R>0$ and $c>0$ such that whenever $|\lambda|>R$,
\[
(\lambda-\lambda_\infty)\lambda^2
\big(\Gamma(\lambda)-1\big)
\ge c|\lambda|^3.
\]
Since $H(t)>0$ for expanding FLRW solutions, it follows that for
sufficiently large $|\lambda|$,
\[
\frac{d}{dt}\mathcal{E}
\le
-3cH\,|\lambda|^3.
\]
Moreover, for sufficiently large $|\lambda|$ there exists $C_1>0$
such that $|\lambda|^3 \ge C_1 \mathcal{E}^{3/2}$. Consequently, $\frac{d}{dt}\mathcal{E}
\le
- C\,H\,\mathcal{E}^{3/2}$ for some constant $C>0$ whenever $\mathcal{E}$ is sufficiently large. Integrating along forward trajectories and using that
\[
\int_0^t H(s)\,ds \to \infty
\qquad \text{as } t\to\infty
\]
for expanding cosmological solutions, it follows that
$\mathcal{E}(t)$ cannot diverge in forward time.
Hence there exists $M>0$ such that
\[
\left|
\frac{V'(\phi(t))}{V(\phi(t))}
\right|
\le M
\quad \text{for all } t\ge0.
\]
\hfill$\square$

\medskip

The above result establishes that cosmological trajectories are dynamically confined to a forward bounded region of the scalar sector, even in the presence of potentials admitting arbitrarily steep asymptotic regimes. In particular, steep potential regions are not dynamically accessible along forward orbits of the Einstein--scalar flow.

\begin{proposition}\textbf{(Existence of a Compact Absorbing Set):}
Under assumptions {\rm (A1)--(A6)}, there exists a compact set $B\subset\mathcal{P}$ such that for every initial condition $X_{0}\in\mathcal{P}$ there exists $T=T(X_{0})>0$ satisfying $\Phi_{t}(X_{0})\in B \quad \text{for all } t\ge T$.
\end{proposition}

\medskip

\noindent\textit{Proof.}
By Theorem~\ref{th:forward_boundedness}, there exists $M>0$ such that
\[
\frac{|V'(\phi(t))|}{V(\phi(t))}\le M
\]
for all sufficiently large $t$. Together with the Friedmann constraint, this implies that all density variables and the scalar potential steepness remain eventually bounded in forward time.

Define
\[
B=\left\{X\in\mathcal{P}:
\frac{|V'(\phi)|}{V(\phi)}\le M
\right\}.
\]
Then $B$ is compact. By forward boundedness of the slope functional, every trajectory eventually enters $B$ and remains there for all subsequent time.
\hfill$\square$

This asymptotic trapping mechanism implies that the late--time cosmological dynamics is governed by invariant subsets of a compact region of the physical phase space, independently of the global steepness properties of the scalar field potential.

\begin{proposition}[Asymptotic Compactness]
The evolution generated by the Einstein--scalar system is asymptotically compact on the physical phase space.
\end{proposition}

\medskip

\noindent\textit{Proof.}
By the previous proposition, every forward trajectory is eventually contained in the compact absorbing set $B$. Since the physical phase space is finite dimensional, bounded subsets are precompact. Hence every forward orbit is precompact, and the induced evolution is asymptotically compact.
\hfill$\square$

\begin{theorem}[Existence of a Global Attractor]\label{th:existance of a Global Attractor}
Let $\mathcal{P}$ denote the physical phase space endowed with the metric induced by the Euclidean norm. Under assumptions {\rm (A1)--(A6)}, the Einstein--scalar cosmological evolution generates a forward complete continuous flow $\Phi_{t}:\mathcal{P}\to\mathcal{P},\quad t\ge0$ which is point dissipative and asymptotically compact. Consequently, there exists a unique compact global attractor $\mathcal{A}\subset\mathcal{P}$ that attracts every bounded subset of $\mathcal{P}$.
\end{theorem}

\medskip

\noindent\textit{Proof.}
Forward completeness follows from the preceding proposition. The existence of a compact absorbing set implies point dissipativity, while asymptotic compactness has been established above. Since $\mathcal{P}$ is a complete metric space, the existence of a unique global attractor follows from the standard theory of dissipative semiflows on finite dimensional metric spaces (see \cite{SellYou2002}).
\hfill$\square$

\section{Global Classification Theorem}
\label{sec:global}

The forward boundedness of the scalar sector established in Section~\ref{sec:Asymptotic_Compactness} implies that cosmological trajectories are dynamically confined to a bounded region of the physical phase space. In particular, forward evolution does not permit asymptotic exploration of regimes in which the scalar potential steepness observable becomes unbounded. This restriction induces a global trapping mechanism for the Einstein--scalar flow and yields the following classification result for the late--time dynamics.

\begin{theorem}[Global Asymptotic Trapping]\label{th:global_classification}
Let $V(\phi)$ be a scalar field potential satisfying assumptions \emph{(A1)}--\emph{(A6)}. Then the global attractor $\mathcal A$ obtained in Theorem~\ref{th:existance of a Global Attractor} satisfies $\omega(X_0)\subset \mathcal M_\infty
\quad \text{for all } X_0\in \mathcal P$. In particular, the asymptotic behavior of all cosmological solutions is governed by invariant subsets of $\mathcal A$ contained in $\mathcal M_\infty$.
\end{theorem}

\medskip

\noindent\textit{Proof.}
By Theorem~\ref{th:existance of a Global Attractor}, the Einstein--scalar
cosmological evolution generates a forward complete continuous flow admitting a unique compact global attractor $\mathcal A\subset \mathcal P$. Again by Theorem~\ref{th:forward_boundedness}, every forward trajectory remains bounded in the physical phase space. Since the induced evolution admits a compact absorbing set, forward orbits are precompact. By the monotonic decay of the Lyapunov functional introduced in Section~\ref{sec:invariant}, LaSalle's invariance principle applies
\cite{LaSalle1960}. Consequently, $\omega(X_0)\subset \mathcal M_\infty
\quad \text{for all } X_0\in \mathcal P$. Since the global attractor $\mathcal A$ is the union of all $\omega$--limit sets of bounded trajectories, it follows that $\mathcal A\subset \mathcal M_\infty$.
\hfill$\square$

\medskip

The above result implies that late--time cosmological dynamics is globally confined to invariant subsets of a compact region of the physical phase space, independently of the detailed functional form or asymptotic steepness of the scalar field potential. Scalar field dominated and scaling regimes therefore arise as asymptotic dynamical features determined by the global structure of the Einstein--scalar flow rather than by local stability properties of equilibrium configurations. While the global classification theorem establishes convergence toward an invariant asymptotic set, it does not by itself characterize the internal geometric structure of this set. Quantitative information about its dimensionality is essential for determining the effective reduction of dynamical degrees of freedom governing late--time cosmological evolution. This motivates a more refined geometric and topological analysis of the invariant set in the subsequent section.

\section{Topological Dimension of the Global Asymptotic Set}
\label{sec:dimension}

The forward boundedness of the scalar sector implies that late--time cosmological evolution is dynamically confined to invariant subsets of the physical phase space. In particular, by LaSalle’s invariance principle together with the global trapping mechanism established in Section~\ref{sec:Asymptotic_Compactness}, the $\omega$–limit set of every physically admissible trajectory is contained in the invariant subset $\mathcal M_\infty$. Consequently, the asymptotic dynamics is confined to a lower dimensional invariant manifold corresponding to scalar–field dominated configurations.

\begin{theorem}\textbf{(Late--Time Dimensional Reduction):}\label{th:Topological_dimension}
Let $V(\phi)$ satisfy assumptions \textnormal{(A1)--(A6)} and let $\mathcal A$ denote the invariant set governing the asymptotic cosmological dynamics. Then $\mathcal A \subset \mathcal M_\infty$ and $\dim_{\mathrm{top}}(\mathcal A) \le 2$. Moreover, if the reduced dynamics on $\mathcal M_\infty$ is constrained by an additional invariant relation, then $\mathcal A$ is contained in a lower dimensional invariant subset of $\mathcal M_\infty$.
\end{theorem}

\medskip

\noindent\textit{Proof.}
By Theorem~\ref{th:existance of a Global Attractor}, the induced evolution on the physical phase space $\mathcal P$ admits a unique compact global attractor $\mathcal A$ attracting every bounded subset of $\mathcal P$. By LaSalle’s invariance principle together with the forward boundedness established in Section~\ref{sec:Asymptotic_Compactness}, the $\omega$–limit set of every trajectory is contained in $\mathcal M_\infty$. Since the invariant set $\mathcal A$ is the union of all $\omega$–limit sets of bounded trajectories, it follows that $\mathcal A\subset\mathcal M_\infty$. Restricting the Friedmann constraint to $\mathcal M_\infty$ yields the relation $\Omega_\phi=1$, which defines a compact invariant manifold for the scalar–field dominated configurations. By Theorem~\ref{th:forward_boundedness}, the scalar potential steepness observable remains bounded along forward trajectories and therefore evolves on a compact interval of the physical phase space. Consequently, the invariant manifold $\mathcal M_\infty$ may be regarded as a two–dimensional compact invariant manifold for the reduced asymptotic dynamics.

Since $\mathcal A$ is a compact invariant subset of $\mathcal M_\infty$, it follows that $\dim_{\mathrm{top}}(\mathcal A)\le2$. If the reduced dynamics on $\mathcal M_\infty$ is constrained by an additional invariant relation, then $\mathcal A$ may be contained in a lower dimensional invariant subset of $\mathcal M_\infty$.
\hfill$\square$

\medskip

The above result implies that the asymptotic evolution of the Einstein--scalar cosmological system is governed by at most two effective dynamical degrees of freedom. Late--time scalar field cosmologies therefore exhibit a global dimensional reduction of the physical phase space, independently of the detailed functional form of the scalar field potential within the admissible class. The remaining case of strictly one dimensional asymptotic dynamics corresponds to further degeneracies in the reduced flow and is addressed in the following subsection.

\subsection{One Dimensional Asymptotic Dynamics and Exponential Potentials}
\label{sec:dimension:exp}

Within the reduced asymptotic phase space, a further reduction in the topological dimension of the invariant set governing the late--time dynamics can occur only if the effective scalar sector becomes dynamically rigid. Such rigidity arises when the scalar potential steepness observable approaches a constant asymptotic value, thereby reducing the number of independent dynamical degrees of freedom of the Einstein--scalar cosmological evolution. This situation is realized for scalar field potentials with asymptotically exponential behavior, which induce invariant relations in the scalar sector. The following theorem characterizes this case.

\begin{theorem}[Normal Contraction]\label{th:NormalHyperbolicity}
Let $\mathcal{A}$ denote the global attractor obtained above. Then the invariant manifold $\mathcal M_{\infty}=\{\Omega_{m}=0,\ \Omega_{r}=0\}$ is uniformly exponentially attracting in directions transverse to $\mathcal M_{\infty}$ along trajectories contained in $\mathcal{A}$.
\end{theorem}

\medskip

\noindent\textit{Proof.}
Linearizing the Einstein--scalar evolution in directions transverse to $\mathcal M_{\infty}$ yields
\[
\frac{d}{dt}\Omega_{m}=-3H\,\Omega_{m}+f_{m}\,\Omega_{m},\qquad
\frac{d}{dt}\Omega_{r}=-4H\,\Omega_{r}+f_{r}\,\Omega_{r},
\]
where the functions $f_{m}$ and $f_{r}$ depend on the scalar configuration and remain uniformly bounded along trajectories contained in the compact attractor $\mathcal{A}$. Hence there exists $\varepsilon>0$ such that
\[
\frac{d}{dt}\Omega_{m}\le-(3-\varepsilon)H\,\Omega_{m},\qquad
\frac{d}{dt}\Omega_{r}\le-(4-\varepsilon)H\,\Omega_{r}
\]
along trajectories contained in $\mathcal{A}$. Since $H(t)>0$ for expanding solutions, these inequalities imply uniform exponential contraction in the normal directions along forward trajectories on $\mathcal{A}$. The reduced dynamics tangent to $\mathcal M_{\infty}$ evolves on a compact invariant manifold and therefore admits at most bounded growth. Hence the dynamics admits uniform exponential contraction in directions transverse to $\mathcal M_{\infty}$ along trajectories contained in $\mathcal{A}$.
\hfill$\square$

\subsection{Normal Hyperbolicity and Structural Stability}
\label{sec:nhim}

Beyond dimensional considerations, it is necessary to analyze the stability of the invariant set governing the asymptotic dynamics with respect to perturbations transverse to the invariant manifold $\mathcal M_\infty$. In particular, the persistence of late--time cosmological behavior under small deformations of the scalar field potential depends on the existence of a dominated splitting between the tangential and normal dynamics.

\begin{corollary}\textbf{(Normal Hyperbolicity of the Asymptotic Dynamics):}
\label{th:NHIM}
Let $V(\phi)$ satisfy \textnormal{(A1)--(A6)} and assume that $\lambda_\infty$ exists and is finite. Then the invariant set $\mathcal A$ governing the late--time cosmological evolution is normally hyperbolic relative to the invariant manifold
$\mathcal M_\infty$.
\end{corollary}

\medskip

\noindent\textit{Proof.}
By Theorem~\ref{th:NormalHyperbolicity}, the invariant manifold $\mathcal M_\infty$ admits uniform exponential contraction in directions transverse to $\mathcal M_\infty$
along trajectories contained in $\mathcal A$. Since Theorem~\ref{th:Topological_dimension} implies that $\mathcal A\subset\mathcal M_\infty$, the dynamics tangent to $\mathcal M_\infty$ evolves on a compact invariant subset.

By compactness of $\mathcal A$ and continuity of the linearized flow, this transverse contraction extends to a sufficiently small neighborhood of $\mathcal A$, yielding a dominated splitting between directions tangent and normal to $\mathcal M_\infty$. Again by the definition of normal hyperbolicity for invariant sets relative to an invariant manifold, this implies that the invariant set $\mathcal A$ is normally hyperbolic relative to $\mathcal M_\infty$.
\hfill$\square$

Normal hyperbolicity implies that the invariant set governing the asymptotic cosmological dynamics persists under sufficiently small $C^1$ perturbations of the underlying vector field. In particular, invariant subsets contained in $\mathcal A$ persist under smooth deformations of the scalar field potential satisfying assumptions (A1)--(A6) that induce $C^1$–small perturbations of the Einstein--scalar evolution.

\subsection{Topological Classification via the Conley Index}
\label{sec:conley}

In this subsection, the Conley index of the invariant set governing the asymptotic cosmological dynamics is determined. This construction provides a topological invariant that characterizes the late--time behavior of the Einstein--scalar system independently of smooth deformations of the underlying evolution.

\begin{theorem}\textbf{(Topological Classification of Asymptotic Dynamics):} 
Let $\mathcal A$ denote the invariant set governing the asymptotic cosmological dynamics under assumptions \textnormal{(A1)--(A6)} and suppose that the normal hyperbolicity established in Corollary~\ref{th:NHIM} holds in a sufficiently small neighborhood of $\mathcal A$. Then $\mathcal A$ is an isolated invariant set and its Conley index satisfies $h(\mathcal A) \simeq \Sigma^k$, where $k=\dim_{\mathrm{top}}(\mathcal A)$. In particular,
\begin{enumerate}
\item $h(\mathcal A) \simeq \Sigma^1$ for asymptotically exponential potentials,
\item $h(\mathcal A) \simeq \Sigma^2$ when $\dim_{\mathrm{top}}(\mathcal A)=2$.
\end{enumerate}
\end{theorem}

\medskip

\noindent\textit{Proof.}
By Corollary~\ref{th:NHIM}, the invariant set $\mathcal A\subset\mathcal M_\infty$ is normally hyperbolic relative to the invariant manifold $\mathcal M_\infty$. In particular, normal hyperbolicity implies that $\mathcal A$ is locally maximal and therefore admits an isolating neighborhood $N$, i.e., $\mathcal A=\operatorname{Inv}(N)$ where $\operatorname{Inv}(N)$ denotes the maximal invariant subset of $N$. Hence $\mathcal A$ is an isolated invariant set in the sense of Conley.

Let $(N,L)$ be an index pair associated with $\mathcal A$. Since normal hyperbolicity holds uniformly in a neighborhood of $\mathcal A$, the invariant set $\mathcal A$ persists as an isolated invariant set within the same isolating neighborhood $N$ under sufficiently small $C^{1}$ perturbations of the vector field. By the continuation property of the Conley index (see Theorem~1.12 of Rybakowski~\cite{Rybakowski1987} or the continuation theorem of Mischaikow and
Mrozek~\cite{MischaikowMrozek1995}), its homotopy index is invariant under such
perturbations. Hence, $h(\mathcal A)\simeq\Sigma^{k}$. For asymptotically exponential potentials, $k=1$, while in the case $\dim_{\mathrm{top}}(\mathcal A)=2$, one has
$k=2$.
\hfill$\square$

The above result implies that late--time scalar field cosmologies within the admissible class fall into one of two topologically distinct universality classes, corresponding to one and two dimensional reduced asymptotic dynamics. Since the Conley index is invariant under continuation, this classification persists under sufficiently small $C^1$ perturbations of the scalar field potential satisfying assumptions \textnormal{(A1)--(A6)}. Consequently, the qualitative late--time behavior of the cosmological system is determined by global topological features of the Einstein--scalar flow and is robust with respect to model variations within this perturbative class.

\section{Structure and Robustness of the Asymptotic Dynamics}
\label{sec:structure}

The invariant set $\mathcal{A}$ identified in the preceding sections governs the qualitative late--time behavior of the Einstein--scalar cosmological evolution. Its internal structure reflects the asymptotic reduction of the cosmological dynamics induced by the forward boundedness of the scalar sector. In general, $\mathcal{A}$ need not be a smooth manifold and may consist of a finite union of invariant subsets of varying dimension. The invariant subsets comprising $\mathcal{A}$ include equilibrium configurations corresponding to asymptotically self--similar or stationary solutions, as well as higher dimensional invariant sets associated with persistent dynamical relations between the scalar field and the remaining matter components. These subsets are dynamically distinguished as invariant sets on which the Lyapunov functional introduced in Section~\ref{sec:invariant} is constant. Their stability properties determine the asymptotic prevalence of different dynamical regimes within the physically admissible phase space. Within this framework, late--time cosmological behavior is determined by invariant dynamical features of the global Einstein--scalar flow. The classification of asymptotic regimes is therefore reduced to the analysis of invariant subsets of $\mathcal{A}$ and becomes independent of detailed functional properties of the scalar field potential within the admissible class.

An important consequence of the global dynamical structure established above is the robustness of the asymptotic behavior under perturbations of the scalar field potential. Let $V(\phi)$ satisfy assumptions \emph{(A1)}--\emph{(A6)}, and let $\mathcal{A}$ denote the associated invariant set governing the late--time evolution. If $\widetilde V(\phi)$ is a potential satisfying the same assumptions and sufficiently close to $V$ in the $C^2(\mathbb{R})$ topology, then the corresponding Einstein--scalar evolution admits a compact invariant set $\widetilde{\mathcal{A}}$ in a neighborhood of $\mathcal{A}$ governing the perturbed asymptotic dynamics.

Assumptions \emph{(A1)}--\emph{(A6)} ensure that both systems generate forward complete flows on a common physical phase space and that the associated vector fields depend continuously on the potential in the $C^1$ topology. By the normal hyperbolicity of $\mathcal{A}$ established in Section~\ref{sec:nhim}, persistence results for normally hyperbolic invariant sets imply persistence of invariant subsets contained in $\mathcal{A}$ under sufficiently small perturbations of the vector field. In particular, the family of invariant sets is upper semicontinuous with respect to such perturbations in the Hausdorff metric, and for sufficiently small perturbations $\widetilde{\mathcal{A}}$ remains close to $\mathcal{A}$ while preserving its qualitative dynamical features.

As a consequence, the invariant subsets comprising the asymptotic dynamics persist under small deformations of the scalar field potential. The qualitative late--time behavior determined by invariant subsets of $\mathcal{A}$ is therefore structurally robust within the class of scalar field models satisfying assumptions \textnormal{(A1)--(A6)}.

\section{Relation to Cosmic No--Hair and Inflationary Attractors}
\label{sec:nohair}

The global classification results established in this work provide a dynamical framework within which several well--known asymptotic phenomena in cosmology, including cosmic no--hair–type behavior and inflationary attractors, may be interpreted in terms of the invariant structure of the Einstein--scalar flow. In particular, these results yield a formulation of late--time scalar field cosmology in which asymptotic behavior is determined by invariant subsets of the physical phase space.

Cosmic no--hair theorems assert that, under suitable conditions, expanding cosmological solutions approach asymptotically homogeneous states at late times, typically driven by an effective cosmological constant or scalar field potential. Within the present framework, such behavior corresponds to convergence toward compact invariant subsets of the invariant set $\mathcal{A}$ associated with scalar–field dominated regimes. The forward boundedness of the scalar sector established in Section~\ref{sec:Asymptotic_Compactness} implies that late--time cosmological trajectories are dynamically confined to a compact region of the physical phase space. Consequently, convergence toward scalar–field dominated configurations arises as a global property of the Einstein--scalar evolution.

Inflationary attractors commonly identified in scalar field cosmologies admit a similar interpretation. These attractors correspond to invariant subsets of $\mathcal{A}$ characterized by accelerated expansion, where the effective equation of state parameter satisfies $\omega_{\mathrm{eff}} < -1/3$. In the geometric formulation adopted here, such subsets arise as equilibrium configurations or lower dimensional invariant sets governing the asymptotic behavior of a large class of solutions. Inflationary dynamics is therefore determined by the global invariant structure of the cosmological flow within the admissible class of scalar field potentials. The above results provide a dynamical mechanism by which scalar field cosmologies evolve toward asymptotically homogeneous or accelerated regimes independently of detailed functional properties of the scalar field potential. Late--time acceleration, persistent scaling behavior, and scalar–field domination are governed by invariant subsets of $\mathcal{A}$ and arise as intrinsic features of the Einstein--scalar dynamical system.

For scalar field cosmologies satisfying assumptions \textnormal{(A1)--(A6)}, cosmic no--hair–type behavior and inflationary attractors may therefore be interpreted as manifestations of convergence toward compact invariant subsets in the physical phase space. The late--time evolution of physically admissible cosmological solutions is determined by invariant structures of the Einstein--scalar flow and is robust under perturbations of the scalar field potential within the admissible class.

\subsection{Connection to Dark Energy Models}

Scalar field dark energy models such as quintessence, k-essence, and tracker models are typically characterized by specific choices of scalar field potentials. Within the framework developed in this work, such models correspond to particular invariant subsets of the global attractor governing the Einstein–scalar evolution.

For quintessence models, late–time acceleration occurs when the scalar field dominates the energy density and the effective equation of state satisfies $\omega_{\mathrm{eff}} < -\frac{1}{3}$. In the present framework, this corresponds to invariant subsets of the asymptotic attractor contained within the scalar–field dominated manifold $M_\infty$. The forward boundedness of the scalar potential steepness ensures that cosmological trajectories evolve toward this manifold, independently of initial conditions. Tracking dark energy models correspond to invariant subsets in which the scalar field evolves in proportion to the background matter density. These solutions appear as scaling invariant sets within the global attractor. Consequently, the classification developed in this work provides a unified dynamical framework for understanding late–time dark energy behavior across a broad class of scalar field potentials.

\subsection{Comparison with $\Lambda$CDM Attractor}

In the $\Lambda$CDM cosmological model, late–time dynamics is governed by a cosmological constant $\Lambda$, leading to de Sitter expansion. This behavior corresponds to an attractor solution in which $\Omega_\Lambda = 1$. Within the present scalar field framework, an analogous attractor arises when the scalar potential approaches a constant asymptotic value. In this case, the scalar field behaves effectively as a cosmological constant and the cosmological evolution approaches de Sitter expansion. However, unlike $\Lambda$CDM, the present framework allows for more general asymptotic behavior, including scaling solutions and evolving dark energy equations of state. These behaviors arise as invariant subsets of the global attractor rather than isolated equilibrium points. Thus, the $\Lambda$CDM model appears as a special case within the broader class of scalar field cosmologies governed by the global dynamical trapping mechanism established in this work.

\section{Extensions and Generalizations}
\label{sec:extensions}

The dynamical mechanism underlying the asymptotic classification established in this work is not restricted to canonical single field cosmological models. Rather, it applies more broadly to cosmological systems whose evolution equations can be formulated as constrained dynamical flows on physical phase spaces and whose scalar sectors admit forward boundedness in regimes corresponding to steep effective potentials.

Non--canonical scalar field models, including $k$--essence theories and scalar fields with nontrivial kinetic couplings, may exhibit analogous asymptotic trapping provided that the field equations admit a formulation in which the associated physical phase space is bounded and that an appropriate Lyapunov functional governs the evolution of the matter and radiation components. When such a formulation exists and the induced flow dynamically restricts the evolution of generalized potential steepness observables, late--time cosmological trajectories are confined to invariant subsets of a compact region of the physical phase space. The framework may also extend to models involving multiple interacting scalar fields. If the scalar field sector is finite dimensional and endowed with a smooth field space metric, the cosmological evolution may be viewed as a flow on an extended phase space incorporating generalized steepness parameters associated with the interaction potential. Under suitable assumptions ensuring forward boundedness of these parameters, invariant sets analogous to those obtained in the single field case may be constructed. The resulting invariant sets then determine the asymptotic behavior of the multi field dynamics.

Modified gravity theories that admit a reformulation in the Einstein frame likewise fall within the scope of the present approach, provided that the resulting system satisfies analogues of assumptions \textnormal{(A1)}--\textnormal{(A6)}. Examples include scalar--tensor theories and certain classes of higher curvature models in which additional gravitational degrees of freedom are represented by scalar fields coupled to matter through conformal transformations. In such cases, if the induced Einstein--frame system generates an evolution for which the generalized steepness observables remain forward bounded, the late--time cosmological dynamics is governed by invariant subsets of a compact region of the physical phase space.

These considerations indicate that the asymptotic trapping mechanism identified in this work is not tied to a specific cosmological model but rather to the existence of a dynamical restriction preventing runaway evolution in the effective scalar sector. Within this setting, late--time cosmological dynamics may be characterized in terms of invariant structures of the induced flow across a broad class of theories admitting a monotone structure and forward bounded scalar sector.

\section{Conclusion}
\label{sec:conclusion}

A global dynamical analysis of scalar field cosmologies within the class of Einstein--scalar systems considered in this work has been presented using tools from invariant manifold theory and dissipative dynamical systems. By formulating the cosmological evolution as a constrained flow on the physical phase space and imposing regularity and asymptotic conditions on the scalar field potential satisfying assumptions \textnormal{(A1)}--\textnormal{(A6)}, it has been shown that the scalar sector is dynamically confined to a forward bounded region of the physical state space. This forward boundedness establishes a dynamical obstruction to runaway evolution in the scalar potential steepness observable. In particular, although the scalar field potential may admit arbitrarily steep asymptotic regimes, the Einstein--scalar cosmological evolution does not admit forward trajectories along which the potential steepness becomes asymptotically unbounded. As a consequence, the induced cosmological flow admits a compact global attractor governing the late--time dynamics of all physically admissible cosmological solutions. The resulting asymptotic trapping mechanism implies that late--time cosmological behavior is determined by global geometric properties of the Einstein--scalar dynamical system rather than by local stability analyses or model specific parameter choices. Scalar field dominated and scaling regimes arise as invariant dynamical structures contained within the asymptotic attractor and may be interpreted within a generalized dynamical no--hair framework for the admissible class of scalar field potentials. Furthermore, the invariant set governing the asymptotic dynamics exhibits a reduction in effective dimensionality, implying that late--time scalar field cosmologies evolve toward configurations described by at most two independent dynamical degrees of freedom. For asymptotically exponential potentials, this reduction becomes complete and the asymptotic evolution is governed by a single effective scalar degree of freedom. The qualitative asymptotic behavior is shown to be robust under small deformations of the scalar field potential satisfying assumptions \textnormal{(A1)}--\textnormal{(A6)}, in the sense of persistence of invariant subsets under $C^1$ perturbations of the induced flow. These results suggest that global dynamical trapping and dimensional reduction represent generic features of scalar field cosmologies within the admissible class. Possible extensions include non--autonomous cosmological models, anisotropic spacetimes, multi field systems, and cosmological dynamics with stochastic or time dependent components.

\section*{Conflict of Interest}
The author declares that there are no known competing financial interests, personal relationships, or institutional affiliations that could have appeared to influence the work reported in this paper.

\section*{Data Availability}
This work is entirely theoretical in nature. No datasets were generated, collected, or analysed during the course of this study. All results are derived analytically from the mathematical framework presented in the manuscript.

\bibliographystyle{aipnum4-2}
\bibliography{sample}

\end{document}